%% file: RR-7490.tex
\documentclass[a4paper]{article}

\usepackage{RRA4}
\usepackage{hyperref}
\usepackage[centertags]{amsmath}
\usepackage{latexsym}
\usepackage{ amsfonts, amssymb}
\usepackage{amsthm}
\usepackage{algorithm}
\usepackage[noend]{algorithmic}
\usepackage{hyperref}

\input{macros}

\usepackage[utf8]{inputenc}
\RRdate{Décembre 2010}

\RRauthor{
Daniel Augot\thanks[sfn]{INRIA Saclay Île-de-France}%
  \and
Morgan Barbier  \thanksref{sfn}
\and 
Alain Couvreur\thanks{Insitut de Mathématiques de Bordeaux, UMR 5251 - Université Bordeaux 1, 351 cours de la Libération, 33405 {\sc Talence} Cedex}%
}
\RRtitle{Décodage en liste des codes de Goppa binaires jusqu'à la borne de Johnson binaire}
\RRetitle{List-decoding of binary Goppa codes up to the binary Johnson bound}
\titlehead{List-decoding of binary Goppa codes}
\RRresume{Nous étudions le décodage en liste des codes alternants,
dont notamment les codes de Goppa classiques.  La considération
majeure est de prendre en compte la taille de l'alphabet, qui influe
sur la capacité de correction, surtout dans le cas de l'alphabet
binaire. Cela revient à comparer la borne de Johnson que nous appelons
générique, à la borne de Johnson que nous appelons $q$-aire, qui prend
en compte la taille $q$ du corps. Cette différence est d'autant
plus sensible que $q$ est petit.

Essentiellement, le cas le plus favorable est celui de l'alphabet
binaire pour lequel on peut augmenter significativement le rayon du décodage en
liste. Et ce, d'autant plus que la distance minimale relative construite du
code alternant binaire est proche de $1/2$.

Bien que le résultat annoncé ici, à savoir le rayon de décodage en
liste des codes de Goppa binaires, soit nouveau, il peut assez
facilement être déduit de sources relativement peu connues
(V.~ Guruswami, R.~M.~Roth and I.~Tal, R~.M.~Roth) et dont les auteurs 
n'ont apparemment pas pensé à aborder les codes de Goppa binaires.
Seul D.~J.~Bernstein
a traité le décodage en liste des codes de
Goppa dans une prépublication.
Les références sont données dans l'introduction.

Nous proposons un contenu autonome, et aussi une analyse de la
complexité de l'algorithme étudié, qui est quadratique en la longueur $n$
du code, si on se tient à distance du rayon relatif de décodage
maximal, et en $\mathcal{O}(n^7)$ pour le rayon de décodage
maximal.}

\RRabstract{We study the list-decoding problem of alternant codes,
with the notable case of classical Goppa codes. The major
consideration here is to take into account the size of the alphabet,
which shows great influence on the list-decoding radius. This amounts
to compare the \emph{generic} Johnson bound to the \emph{$q$-ary}
Johnson bound. This difference is important when $q$ is very small.

Essentially, the most favourable case is $q=2$, for which the
decoding radius is greatly improved, notably when the relative minimum
distance gets close to $1/2$.

Even though the announced result, which is the list-decoding radius of
binary Goppa codes, is new, it can be rather easily made up from
previous sources (V.~ Guruswami, R.~M.~Roth and I.~Tal, R~.M.~Roth),
which may be a little bit unknown, and in which the case
of binary Goppa codes has apparently not been thought at. Only D.~J.~Bernstein
treats the case of binary Goppa codes in a preprint. References are given in the introduction.

We propose an autonomous treatment and also a complexity analysis of
the studied algorithm, which is quadratic in the blocklength $n$, when decoding at some
distance of the relative maximum decoding radius, and in $\mathcal{O}(n^7)$ when
reaching the maximum radius.}
\RRmotcle{Codes correcteurs d'erreur, codes géométriques,
décodage en liste, codes alternants, codes de Goppa binaires}
\RRkeyword{ Error correcting codes, algebraic geometric codes, list-decoding,
alternant codes, binary Goppa codes }
\RRprojets{TANC}
\RRdomaine{2} 
\RCSaclay 
\begin{document}
\RRNo{7490}
\makeRR
\input{body}

\end{document}

%% file: macros.tex
\theoremstyle{plain}
\newtheorem{thm}{Theorem}[section]
\newtheorem{lem}[thm]{Lemma}
\newtheorem{prop}[thm]{Proposition}
\newtheorem{cor}[thm]{Corollary}

\theoremstyle{definition}
\newtheorem{defn}[thm]{Definition}
\newtheorem{prob}[thm]{Problem}

\theoremstyle{remark}
\newtheorem{rem}[thm]{Remark}

\newcommand{\deff}[1]{\emph{#1}}

\newcommand{\ceiling}[1]{\ensuremath{\left\lceil#1\right\rceil}}
\newcommand{\ev}{\operatorname{ev}}
\newcommand{\mult}{\operatorname{mult}}

\def\N{{\mathbb N}}

\def\Z{{\mathbb Z}}
\def\P{{\mathbb P}}

\def\F{\mathbb{F}}
\def\G{\mathcal{G}}
\def\D{\mathcal{D}}
\def\E{\mathcal{E}}
\def\A{\mathcal{A}}

\def\D{\mathcal{D}}
\def\supp{\textrm{Supp}}
\def\code{\mathcal{C}}
\def\genus{g(C)}
\def\egaledef{\triangleq}
\def\dimGRS{k_{GRS}}
\def\deltaGRS{\delta_{GRS}}
\def\dGRS{d_{GRS}}
\def\bigO{\mathcal{O}}
\def\wdeg{\operatorname{wdeg}}
\def\pseudoR{R'}
\newcommand{\intervalle}[2]{[#1,#2]}

%% file: body.tex
\section{Introduction}
\input{Introduction}

\section{List-decoding}\label{ListDecoding}
\input{ListDecoding}

\section{Classical Goppa Codes}\label{Alternant}
\input{Traditional}

\section{List-decoding of classical Goppa Codes as Algebraic Geometric codes}\label{SecGoppaAG}
\input{AlgGeoGoppa}

\input{DecodeAlgGeoGoppa}


\section{List decoding of classical Goppa codes as evalu\-ation
  codes}\label{Sec:AlgoClassic}\label{DecodeClassicGoppa}

\input{DecodeClassicGoppa}

\appendix

\input{Appendix}
\bibliographystyle{alpha}
\bibliography{biblio}

%% file: Introduction.tex
In 1997, Sudan presented the first list-decoding algorithm
for Reed-Solomon codes~\cite{SudanAlgo} having a low, yet positive,
rate. Since the correction radius of Sudan's algorithm for these codes
is larger than the one obtained by unambiguous decoding algorithms,
this represented an important milestone in list-decoding, which was
previously studied at a theoretical level.
See~\cite{Elias:IEEE_IT1991} and references therein for considerations
on the ``capacity'' of list-decoding.
Afterwards, Guruswami and Sudan
improved the previous algorithm by adding a multiplicity constraint in
the interpolation procedure.
These additional constraints enable to increase the correction
radius of Sudan's algorithm for Reed-Solomon codes of any
rate~\cite{GSListDecoding}. The number of errors that this algorithm
is able to list-decode corresponds to the \emph{Johnson radius}
$e_\infty(n,d)=\ceiling{n-\sqrt{n(n-d)}}-1$, where $d$ is the
minimum distance of the code.

Actually, when the size $q$ of the alphabet is properly taken into
account $q$, the bound is improved up to
\[
e_q(n,d)=\ceiling{\theta_q\left(n-\sqrt{n\left(n-\frac d{\theta_q}\right)}\right)}-1,
\]
where $\theta_q=1-\frac1q$. See~\cite[Chapter 3]{Guruswami:ARILD2007} for a complete discussion about
these kinds of bounds, relating the list-decoding radius to the
minimum distance. Dividing by $n$, and taking relative values, with
$\delta=\frac dn$, we define $\tau_\infty(\delta)=\frac{e_\infty(n,d)}{n}$, and
$\tau_q(\delta)=\frac{e_q(n,d)}{n}$, which are
\begin{equation}
\tau_\infty(\delta)=1-\sqrt{1-\delta}, \quad 
\tau_q(\delta)=\theta_q\left(1-\sqrt{1-\frac {\delta}{\theta_q}}\right)
\end{equation}
Note that $\tau_q(\delta)$ gets decreasingly close to
$\tau_\infty(\delta)$ when $q$ grows, and that $\tau_2(n,q)$ is the
largest, see Figure~\ref{Fig:ErrorCapacityVaryingq}. We call
$\tau_\infty(\delta)$ the \emph{generic Johnson bound}, which does not
take into account the size of the field, and indeed works over any
field, finite or not. We refer to $\tau_q(\delta)$ as the $q$-ary
Johnson bound, where the influence of $q$ is properly reflected.

\begin{figure}[ht]
  \begin{center} \includegraphics[width=10cm]{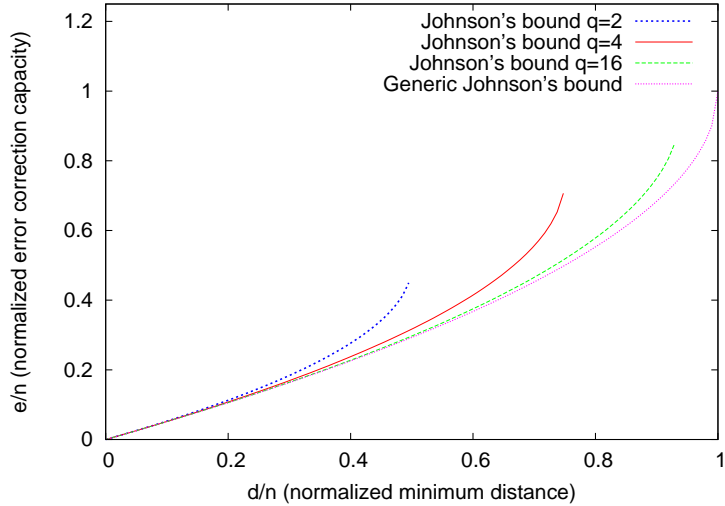}
    \caption{\label{Fig:ErrorCapacityVaryingq} Comparison of the limit
    generic Johnson bound $\tau_\infty(\delta)$ and the limit $q$-ary
    Johnson bounds $\tau_q(\delta)$, for small $q$. Note that the each
    curve ends at $\delta=\theta_q=1-\frac1q$, which is the maximum
    relative minimum distance of codes of positive rates over $\F_q$,
    from the Plotkin bound.} \end{center}
\end{figure}

The truth is that the $\tau_q(\delta)$ radius can be reached for the
whole class of alternant codes, and this paper presents how to do
this. We have essentially compiled existing, but not very well-known
results, with the spirit of giving a report on the issue of list-decoding
classical algebraic codes over bounded alphabets. First, we
have to properly give credits.

Considering the possibility of varying multiplicities, Koetter and
Vardy proposed in 2001, an algebraic soft-decision decoding algorithm
for Reed-Solomon codes~\cite{Koetter}. This method is based on an
interpolation procedure which is similar to Guruswami-Sudan's
algorithm, except that the set of interpolation points is two
dimensional, and may present varying multiplicities, according the
reliability measurements given by the channel. Note that the idea of
varying multiplicities was also considered in~\cite{GSListDecoding}, as
the ``weighted polynomial reconstruction problem'', but was not
instantiated for particular cases, as it was done by Koetter and
Vardy. Before the publication of~\cite{Koetter}, also circulated a
preprint of Koetter and Vardy~\cite{Koetter-Vardy:PREPRINT2000}, which
was a greatly extended version of~\cite{Koetter}, with many possible
interesting instances of the weighted interpolation considered. In
particular, the authors discussed the decoding of BCH codes over the binary
symmetric channel, and reached in fact an error capacity which is
nothing else than $\tau_2(\delta)$. Note that BCH codes are nothing else
than alternant codes, with benefits when the alphabet is $\F_2$. This was
not published.

Guruswami-Sudan's algorithm is in fact very general and can also
be applied to (one point) Algebraic Geometric codes as also shown by
Guruswami and Sudan in \cite{GSListDecoding}.
By this manner, one also reaches the Johnson
radius $\ceiling{n-\sqrt{n(n-d^\star)}}-1$, where $d^\star$ is the
Goppa designed distance. Contrarily to Reed-Solomon codes, it is
possible, for a fixed alphabet $\F_q$, to construct Algebraic Geometric
codes of any length. In this context, it makes sense to try to reach
the $q$-ary Johnson bound $\tau_q(\delta)$, which is done in
Guruswami's thesis~\cite{GuruswamiPhD}, at the end of Chapter 6.

Apparently independently, Roth and Tal considered the list-decoding
problem in~\cite{RothTalISIT}, but only an one page
abstract. Roth's book~\cite{Roth:ITCT2006}, where
many algebraic codes are
presented through the prism of \emph{alternant
codes}, considers the list-decoding of these codes and shows how
to reach the $q$-ary Johnson radius $\tau_q(\delta)$, where
$\delta$ is the minimum distance of the Generalised Reed-Solomon
code from which the alternant code is built. Note that alternant
codes were considered in~\cite{GSListDecoding}, but only the generic
Johnson Bound $\tau_\infty(\delta)$ was discussed there.

Among the alternant codes, the \emph{binary Goppa codes} are
particularly important. They are not to be confused with Goppa's Algebraic
Geometric codes, although there is a strong connection which
is developed in Section~\ref{SecGoppaAG}. These codes are constructed
with a \emph{Goppa polynomial} $G(X)$ of degree $r$ and if this
polynomial is square-free, then the distance of these codes is at
least $2r+1$ which is almost the double of $r$, which is what would be
expected for a generic alternant code. In fact, using the statements
in~\cite{Roth:ITCT2006}, and using the fact that the Goppa code built
with $G(X)$ is the same as the Goppa code built with $G(X)^2$, it is
explicit that these codes can be list-decoded up to the radius
\begin{equation}
\ceiling{\frac 12\left(n-\sqrt{n\left(n-(4r+2)\right)}\right)}-1.
\end{equation}
But actually, the first author who really considered the list-decoding of
binary Goppa codes is D.~J.~Bernstein~\cite{BernsteinLDGC}, in a
preprint which can be found on his personal web page. He uses a
completely different approach than interpolation based list-decoding
algorithms, starting with Patterson's
algorithm~\cite{Patterson} for decoding classical Goppa
codes. Patterson's algorithm is designed to decode up to $t$ errors,
and to list-decode further, Bernstein reuses the information obtained
by an unsuccessful application of Patterson's algorithm in a smart
way. It is also the approach used by Wu~\cite{Wu:IEEE_IT2008} in his
algorithm for list-decoding Reed-Solomon and BCH codes, where the
Berlekamp-Massey algorithm is considered instead of Patterson's
algorithm. Notice that Wu can reach the binary Johnson bound
$\tau_2(\delta)$, using very particular properties of Berlekamp-Massey
algorithm for decoding \emph{binary} BCH
codes~\cite{Berlekamp:AGT1968,Berlekamp:ACTRE}. However, Wu's
approach can apparently not be straightforwardly applied to Goppa codes.

\paragraph{Organisation of the paper}
Section~\ref{ListDecoding} is devoted to recall the list-decoding
problem, the Johnson bounds generic or $q$-ary, and
Section~\ref{Alternant} to the definitions of the codes we wish to
decode, namely alternant codes and classical Goppa
codes. Section~\ref{SecGoppaAG} shows how to consider classical Goppa
codes as subfield subcodes of Algebraic Geometric Goppa codes. Then,
using Guruswami's result in~\cite{GuruswamiPhD}, it is
almost straightforward to show that these codes can be decoded up to
the binary Johnson bound $e_2(n,d^\star)$, where $d^\star=2r+1$. However,
this approach is far reaching, and the reader may skip
Section~\ref{SecGoppaAG}, since Section~\ref{DecodeClassicGoppa}
provides a self-contained treatment of the decoding of alternant
codes up to the $q$-ary Johnson bound. Essentially, this amounts to
show how to pick the varying multiplicities, but we also study the
dependency on the multiplicity. This enables us to give an estimation
of the complexity of the decoding algorithm, which is quadratic in the
length $n$ of the code, when one is not too greedy.

%% file: ListDecoding.tex
First,
recall the notion of list-decoding and multiplicity.
\begin{prob}
  Let $\code$ be a code in its ambient space $\mathbb{F}_q^n$. The
  list-decoding problem of $\code$ up to the radius
  $e\in\intervalle0n$ consists, for any $y$ in $\mathbb{F}_q^n$, in
  finding all the codewords $c$ in $\code$ such that $d(c,y)\le e$.
\end{prob}
The main question is: how large can $e$ be, such that the list keeps a
reasonable size? A partial answer is given by the so-called Johnson
bound.

%% file: Traditional.tex
This section is devoted to the study of classical $q$--ary Goppa
codes, regarded as alternant codes (subfield subcodes of Generalised
Reed--Solomon codes) and as subfield subcodes of Algebraic Geometric
codes.  Afterwards, using Guruswami's results \cite{GuruswamiPhD}
on soft-decoding of Algebraic Geometric codes, we prove that
classical Goppa codes can be list-decoded in polynomial time up to the
$q$--ary Johnson bound.

\paragraph{Context}
In this section, $q$ denotes an arbitrary prime power and $m, n$
denote two positive integers such that $m\geq 2$ and $n\leq q^m$.  In
addition $L\egaledef (\alpha_1, \ldots , \alpha_n)$ denotes an
$n$--tuple of distinct elements of $\F_{q^m}$.

\subsection{Classical  Goppa codes}

\begin{defn}
Let $r$ be an integer such that $0<r<n$.  Let $G\in \F_{q^m}[X]$ be a
polynomial of degree $r$ which does not vanish at any element of $L$.
The $q$--ary classical Goppa code $\Gamma_q (L,G)$ is defined by
$$
\Gamma_q (L,G)\egaledef \left\{(c_1, \ldots, c_n)\in \F_q^n\ \left|\ \sum_{i=1}^n \frac{c_i}{X-\alpha_i}\equiv 0\ \textrm{mod}\ (G(X)) \right.  \right\} \cdot
$$  
\end{defn}

\subsection{Classical Goppa codes are alternant}
\begin{defn}[Evaluation map]\label{EvaluationMap}
Let $B\egaledef (\beta_1, \ldots , \beta_n)$ be an $n$--tuple of
elements of $\F_{q^m}^{\times}$, and $L\egaledef (\alpha_1, \ldots ,
\alpha_n)$ denotes an $n$--tuple of distinct elements of $\F_{q^m}$. The associated \deff{evaluation map} is:
$$
\ev: \left\{ \begin{array}{ccc}
\F_{q^m}[X] & \rightarrow & \F_{q^m}^ n\\
f(X)&\mapsto & (\beta_1 f(\alpha_1), \ldots , \beta_n f(\alpha_n))
\end{array}\right. .
$$
\end{defn}

\begin{defn}[Generalised Reed--Solomon code]\label{GRS}
Let $B\egaledef (\beta_1, \ldots , \beta_n)$ be an $n$--tuple of
elements of $\F_{q^m}^{\times}$. Let $k$ be a positive integer. The
\emph{Generalised Reed--Solomon} code (or GRS code) over $\F_{q^m}$
associated to the triple $(L,B,k)$ is the code:
$$
GRS_{q^m}(L,B,k)\egaledef 
\left\{ \ev(f(X)) \ | \ f\in \F_{q^m}[X]_{<k} \right\},
$$ 
where $\ev$ denotes the evaluation map in Definition \ref{EvaluationMap}.
This code has parameters $[n,k,n-k+1]_{q^m}$ (\cite{SlMcW} $Ch 10, \S
8, p303$
).
\end{defn}

\begin{defn}[Subfield Subcode]
  Let $K$ be a finite field and $M/K$ be a finite extension of it.
Let $\code$ be a code of length $n$ with coordinates in $M$, the \emph{subfield subcode} $\code_{|K}$ of $\code$ is the code
$$
\code_{|K}\egaledef \code\cap K^n.
$$
\end{defn}

\begin{defn}[Alternant code]
A code is said to be \emph{alternant} if it is a subfield subcode of a GRS code.  
\end{defn}

In particular, classical $q$--ary Goppa codes are alternant.
Let us describe a GRS code over $\F_{q^m}$ whose subfield subcode over $\F_q$ is $\Gamma_q(L,G)$.

\begin{prop}\label{GoppaGRS}
Let $r$ be an integer such that $0<r<n$ and $G\in \F_{q^m}[X]$ be a polynomial of degree $r$ which does not vanish at any element of $L$.
Then, the classical Goppa code $\Gamma_q (L,G)$ is the subfield subcode $GRS_{q^m}(L,B,n-r)_{|\F_q}$, where $B=(\beta_1, \ldots, \beta_n)$ is defined by
$$
\forall i \in \{1, \ldots , n\},\ \beta_i\egaledef \frac{G(\alpha_i)}{\prod_{j \neq   i} (\alpha_i-\alpha_j)} \cdot
$$
\end{prop}

\begin{proof}
See \cite{SlMcW} $Ch 12, \S 3 , p340, Thm 4$.

\end{proof}

\subsection{A property on the minimum distance of classical Goppa codes}
Let $L\egaledef (\alpha_1, \ldots , \alpha_n)$ be an $n$--tuple of distinct elements of $\F_{q^m}$ and $G\in \F_{q^m}[X]$ be a polynomial of degree $r>0$ which does not vanish at any element of $L$.
Since $\Gamma_q (L,G)$ is the subfield subcode of a GRS code with
parameters $[n, n-r, r+1]_{q^m}$, the code $\Gamma_q (L,G)$ has
parameters $[n,\geq n-mr,\geq r+1]_q$ (see \cite{sticht} Lemma
VIII.1.3 and \cite{SlMcW} $Ch 12, \S 3, p339$).

In addition, it is possible to get a better estimate of the minimum distance in some situations. This is the objective of the following result.

\begin{thm}\label{Egalite_generalisee}
Let $L\egaledef (\alpha_1, \ldots , \alpha_n)$ be an $n$--tuple of distinct elements of $\F_{q^m}$. 
Let $G\in \F_{q^m}[X]$ be square-free polynomial which does not vanish at any element of $L$ and such that $0<\deg (G)< n/q$.
Then,
$$
\Gamma_q (L,G^{q-1})=\Gamma_q (L,G^q).
$$
\end{thm}

\begin{proof}
\cite{BernLangPeters} Theorem 4.1.
\end{proof}

The codes $\Gamma_q(L, G^{q-1})$ and $\Gamma_q (L, G^q)$ are subfield subcodes of two distinct GRS codes but are equal.
The GRS code associated to $G^{q-1}$ has a larger dimension than the one associated to $G^q$ but a smaller minimum distance. Thus, it is interesting to deduce a lower bound for the minimum distance from the GRS code associated to $G^q$ and a lower bound for the dimension from the one associated to $G^{q-1}$.

This motivates the following definition.

\begin{defn}\label{designed}
  In the context of Theorem \ref{Egalite_generalisee}, the designed minimum distance of $\Gamma_q (L,G^{q-1})$ is $d_{\textrm{Gop}}^{\star}\egaledef q\deg(G)+1$. It is a lower bound for the actual minimum distance.
\end{defn}

\begin{rem}\label{generalcase}
  Using almost the same proof, Theorem \ref{Egalite_generalisee} can be generalised as: let $G_1, \ldots , G_t$ be irreducible polynomials in $\F_{q^m}[X]$ and $e_1, \ldots , e_t$ are positive integers congruent to $-1 \mod q$, then
$
\Gamma_q (L, G_1^{e_1}\cdots G_t^{e_t})=\Gamma_q (L, G_1^{e_1+1}\cdots G_t^{e_t+1})
$.
\end{rem}


%% file: AlgGeoGoppa.tex
In this section, $C$ denotes a smooth projective absolutely irreducible curve over a finite field $\F_q$.
Its genus is denoted by $\genus$.

\subsection{Prerequisites in algebraic geometry}
The main notions of algebraic geometry used in this article are summarised in what follows. For further details, we refer the reader to \cite{fulton} for theoretical results on algebraic curves and to \cite{sticht} and \cite{NTV} for classical notions on Algebraic Geometric codes.

\paragraph{Points and divisors}
If $k$ is an algebraic extension of the base field $\F_q^m$, we denote by $C(k)$ the set of $k$--rational points of $C$, that is the set of points whose coordinates are in $k$.

The group of divisors $\textrm{Div}_{\overline{\F}_q}(C)$ of $C$ is the free abelian group generated by the geometric points of $C$ (i.e. by $C(\overline{\F}_q)$).
Elements of
$\textrm{Div}_{\overline{\F}_q}(C)$ are of the form
$\G=\sum_{P\in C(\overline{\F}_q)} a_P P$, where the $a_P$'s are integers and are all zero but a finite number of them.
The support of $\G=\sum a_P P$ is the finite set
$$\supp (\G)\egaledef \{P\in C(\overline{\F}_q) \ |\ a_P\neq 0\}.$$

The group $\textrm{Div}_{\F_q}(C)$ of $\F_q$--rational divisors is the subgroup of $\textrm{Div}_{\overline{\F}_q}(C)$ of divisors which are fixed by the Frobenius map.

A partial order is defined on divisors:
$$\D=\sum d_P P\geq \E=\sum e_P P \ \Leftrightarrow \ \forall P\in C(\overline{\F}_q),\ d_P \geq e_P.$$
A divisor $\D=\sum d_P P$ is said to be \emph{effective} or \emph{positive} if $\D\geq 0$, i.e. if for all $P\in C(\overline{\F}_q)$, $d_P\geq 0$.
To each divisor $\D=\sum_{P}d_P P$, we associate its degree $\deg (\D)\in \Z$ defined by $\deg(\D)\egaledef \sum d_P$. This sum makes sense since the $d_P$'s are all zero but a finite number of them. 

\paragraph{Rational functions}
The field of $\F_q$--rational functions on $C$ is denoted by $\F_q(C)$.
For a nonzero function $f\in \F_q(C)$, we associate its divisor 
$$
(f)\egaledef \sum_{P\in C(\overline{\F}_q)} v_P (f).P,
$$
where $v_P$ denotes the valuation at $P$.
This sum is actually finite since the number of zeroes and poles of a function is finite.
Such a divisor is called a \emph{principal divisor}.
The positive part of $(f)$ is called \emph{the divisor of the zeroes} of $f$ and denoted by
$$
{(f)}_0\egaledef \sum_{
    P\in C(\overline{\F}_q),\
 v_P(f)>0
} v_P(f).P .
$$

\begin{lem}\label{degrePrin}
  The degree of a principal divisor is zero.
\end{lem}

\begin{proof}
  \cite{fulton} Chapter 8 Proposition 1.
\end{proof}


\paragraph{Riemann--Roch spaces}

Given a divisor $\G\in \textrm{Div}_{\F_q}(C)$, one associates a vector space of rational functions defined by
$
  L(\G)  \egaledef   \{f\in \F_q(C)\ | \ (f)\geq -\G\}\cup \{0\}.
$
This space is finite dimensional and its dimension is bounded below from the Riemann--Roch theorem
$
\dim (L(\G))\geq \deg (\G)+1-\genus.
$
This inequality becomes an equality if $\deg (\G)>2\genus-2$.

\subsection{Construction and parameters of Algebraic Geometric codes}

\begin{defn}\label{AGcode}
  Let $\G$ be an $\F_q$--rational divisor on $C$ and $P_1, \ldots ,
  P_n$ be a set of distinct rational points of $C$ avoiding the
  support of $\G$. Denote by $\D$ the divisor $\D\egaledef P_1+\cdots
  +P_n$. The code $\code_L (\D,\G)$ is the image of the map

$$
\textrm{ev}_{\D}:\left\{
  \begin{array}{ccc}
    L(\G) & \rightarrow & \F_q^n \\
    f & \mapsto & (f(P_1), \ldots , f(P_n))
  \end{array}
\right. .
$$
\end{defn}

The parameters of Algebraic Geometric codes (or AG codes) can be
estimated using the Riemann--Roch theorem and Lemma \ref{degrePrin}.

\begin{prop}\label{AGparam}
  In the context of Definition \ref{AGcode}, assume that $\deg
  (\G)<n$. Then the code $\code_L(\D,\G)$ has parameters $[n,k,d]_q$ where
  $k\geq \deg (\G)+1-\genus$ and $d\geq n-\deg (\G)$. Moreover, if $2\genus-2
  <\deg (\G)$, then $k=\deg (\G)+1-\genus$.
\end{prop}

\begin{proof}
  \cite{sticht} Proposition II.2.2 and Corollary II.2.3.
\end{proof}

\begin{defn}[Designed distance of an AG code]
  The designed distance of $\code_L (\D,\G)$ is
  $d_{\textrm{AG}}^{\star}\egaledef n-\deg (\G)$.
\end{defn}

%% file: DecodeAlgGeoGoppa.tex
\subsection{Classical Goppa codes as Algebraic Geometric codes}
In general, one can prove that the GRS codes are the AG codes on the
projective line (see \cite{sticht} Proposition II.3.5).  Therefore,
from Proposition \ref{GoppaGRS}, classical Goppa codes are subfield
subcodes of AG codes.  In what follows we give an explicit description
of the divisors used to construct a classical Goppa code $\Gamma_q
(L,G)$ as a subfield subcode of an AG code.

\paragraph{Context} The context is that of Section \ref{Alternant}. In
addition, $P_1, \ldots , P_n$ are the points of $\P^1$ of respective
coordinates $(\alpha_1:1),\ldots , (\alpha_n:1)$ and $P_{\infty}$ is
the point ``at infinity'' of coordinates $(1:0)$.
We denote by $\D$ the divisor $\D\egaledef P_1+\cdots +P_n \in \textrm{Div}_{\F_{q^m}}(\P^1)$.
Finally, we set
\begin{equation}\label{phi}
F(X)\egaledef \prod_{i=1}^n (X-\alpha_i)\ \in \F_{q^m}[X].
\end{equation}
   
\begin{rem}\label{Poly}
  A polynomial $H\in \F_q [X]$ of degree $d$ can be regarded as a rational function on $\P^1$ having a single pole at $P_{\infty}$ with multiplicity $d$.
In particular $\deg(H) \leq d \Leftrightarrow (H)\geq -dP_{\infty} \Leftrightarrow H\in L(-dP_{\infty})$.
\end{rem}

\begin{thm}\label{GoppaToAG}
Let $G\in \F_{q^m}[X]$ be a polynomial of degree $r$ such that $0<r<n$.
Then,
$$
\Gamma_q (L,G)=\code_L (\D, \A-\E)_{|\F_q},
$$
where $\A, \E$ are positive divisors defined by $\E\egaledef {(G)}_0$ and $\A\egaledef {(F')}+(n-1)P_{\infty}$, where $F'$ denotes the derivative of $F$.
\end{thm}

\begin{rem}
  The above result is actually proved in \cite{sticht} (Proposition II.3.11) but using another description of classical Goppa codes (based on their parity--check matrices).
Therefore, we chose to give another proof corresponding better to the present description of classical Goppa codes.
\end{rem}

\begin{proof}[Proof of Theorem \ref{GoppaToAG}]
First, let us prove that $\A$ is well-defined and positive.
  Since $F$ has simple roots (see (\ref{phi})), it is not a $p$--th power in $\F_{q^m}[X]$ (where $p$ denotes the characteristic). Thus, $F'$ is nonzero. Moreover $F'$ has degree $\leq n-1$ (with equality if and only if $n-1$ is prime to the characteristic). Remark \ref{Poly} entails $(F') \geq -(n-1)P_{\infty}$ and $\A = (F') +(n-1)P_{\infty}\geq 0$.

\medbreak 

Let us prove the result. Thanks to Proposition \ref{GoppaGRS}, it is sufficient to prove that
$\code_L(\D, \A-\E)= GRS_{q^m} (L,B,n-r)$, where $B=(\beta_1, \ldots , \beta_n)$ with
\begin{equation}\label{beta}
\forall i \in \{1, \ldots , n\},\ \beta_i\egaledef \frac{G(\alpha_i)}{\prod_{j\neq i} (\alpha_i - \alpha_j)} \cdot
\end{equation}
Notice that, 
\begin{equation}\label{phiprime}
\forall i \in \{1, \ldots , n\},\ F'(\alpha_i)=\prod_{j\neq i} (\alpha_i - \alpha_j).
\end{equation}
Let $H$ be a polynomial in $\F_{q^m}[X]_{<n-r}$.
Remark \ref{Poly} yields
$(H)\geq -(n-r-1)P_{\infty}$ and
$$
\begin{array}{rcl}
  \left( \frac{\displaystyle GH}{\displaystyle F'}\right) =(G)+(H)-(F') & \geq & (\E-rP_{\infty}) -(n-r-1)P_{\infty} - \A +(n-1)P_{\infty}\\
  & \geq & -(\A-\E).
\end{array}
$$
Thus, $GH/F' \in L(\A-\E)$ and, from
(\ref{beta}) and (\ref{phiprime}), we have
$$
\frac{GH}{F'}(\alpha_i)=\beta_i H(\alpha_i).
$$
This yields $GRS_{q^m}(L,B,n-r) \subseteq \code_L(\D, \A-\E)$.

\medbreak

For the reverse inclusion, we prove that both codes have the same dimension.
The dimension of $GRS_{q^m}(L,B,n-r)$ is $n-r$.
For $\code_L(\D, \A-\E)$, we first compute $\deg (\A-\E)$. By
definition, $\deg (\A)=\deg ((F'))+n-1$ which equals $n-1$ from
Lemma \ref{degrePrin}.
The degree of $\E$ is that of the polynomial $G$, that is $r$.
Thus, $\deg (\A-\E)=n-1-r$. Since $r$ is assumed to satisfy $0<r<n$ and since the genus of $\P^1$ is zero, we have
$
2\genus-2=-2 < \deg (\A-\E)<n.
$
Finally, Proposition \ref{AGparam} entails $\dim \code_L (\D,
\A-\E)=\deg (\A-\E)+1-\genus = n-r$, which concludes the proof.
\end{proof}

\begin{rem}
  Another and in some sense more natural way to describe $\Gamma_q (L,G)$ as a subfield subcode of an AG code is to use differential forms on $\P^1$. 
By this way, one proves easily that $\Gamma_q (L,G)=\code_{\Omega} (\D, \E-P)_{|\F_q}$.
Then, considering the differential form $\nu\egaledef
\frac{dF}{F}$, one can prove that its divisor is
$(\nu)=\A-\D-P_{\infty}$. Afterwards, using \cite{sticht} Proposition
II.2.10, we get 
$$
\code_{\Omega}(\D,\E-P)=\code_L(\D,
(\nu)-\E+P+\D)=\code_L(\D,\A-\E).
$$
\end{rem}

The main tool for the proof of the list-decodability of classical
Goppa codes is a Theorem on the soft-decoding of AG codes, this is the reason why we
introduce our list-decoding algorithm by an Algebraic Geometric codes
point of view.

\subsection{List-decoding up to the $q$-ary Johnson radius}

\begin{thm}[\cite{GuruswamiPhD} Theorem 6.41]\label{soft}
For every $q$-ary AG-code $\code$ of block-length $n$ and
designed minimum distance $d^{\star}=n-\alpha$, there exists a
representation of the code of size polynomial in $n$ under which the
following holds.
Let $\varepsilon > 0$ be an arbitrary constant.
For $1 \leq i \leq n$ and $\lambda \in \F_q$ , let $w_{i, \lambda}$ be a non-negative real.
Then one can find in $poly(n, q, 1/\varepsilon)$ time, a list of all codewords $c = (c_1 , c_2 , \ldots , c_n)$ of $C$ that satisfy
\begin{equation}\label{softeq}\tag{$\star$}
\sum_{i=1}^n w_{i,c_i} \geq \sqrt{(n-d^{\star})\sum_{i=1}^n \sum_{\lambda \in \F_q} w_{i,\lambda}^2} +\varepsilon \max_{i, \lambda} w_{i,\lambda}.
\end{equation}
\end{thm}

Using this result we are able to prove the following statement.

\begin{thm}\label{list}
  In the context of Theorem \ref{Egalite_generalisee}, the code $\Gamma_q (L,G^{q-1})$ can be list-decoded in polynomial time provided the number of errors $t$ satisfies
$$
t< n\left( \frac{q-1}{q} \right) \left(1- \sqrt{1-\frac{q}{q-1}\cdot \frac{d_{\textrm{Gop}}^{\star}}{n}} \right),
$$
where $d_{\textrm{Gop}}^{\star}\egaledef q\deg (G)+1$ (see Definition \ref{designed}).
That is, the code can be list-decoded up to the $q$--ary Johnson bound associated to the best determination of the minimum distance.
\end{thm}

\begin{proof}
Set $\E\egaledef {(G)}_0$ and 
let $\D=P_1+\cdots +P_n$ be as in \S \ref{SecGoppaAG}.
From Theorem \ref{Egalite_generalisee} together with Theorem
\ref{GoppaToAG}, we have
$$\Gamma_q (L, G^{q-1})=\code_L(\D,
\A-(q-1)\E)_{|\F_q}=\code_L (\D, \A-q\E)_{|\F_q},$$
where $\A$ is as in the statement of Theorem \ref{GoppaToAG}.
We will apply Theorem \ref{soft} to $\code_L (\D, \A-q\E)$.
From Proposition \ref{AGparam}, the designed distance of this code is $d_{\textrm{AG}}^{\star}\egaledef n-\deg (\A)+q\deg (\E)$. 
Since $\deg (\A)=n-1$ and $\deg (\E)=\deg (G)$, we get
$$d_{\textrm{AG}}^{\star}=q\deg (G)+1=d_{\textrm{Gop}}^{\star}.$$

\noindent Let $\delta$ and $\tau$ be respectively the normalised designed distance and expected error rate:
\begin{equation}\label{deltatau}
\delta\egaledef  \frac{q\deg (G)+1}{n} \quad \textrm{and} \quad
\tau\egaledef \left( \frac{q-1}{q} \right) \left(1- \sqrt{1-\frac{q\delta}{q-1}} \right)\cdot 
\end{equation}

The approach is almost the same as that of \cite{GuruswamiPhD} \S 6.3.8.
Assume we have received a word $y \in \F_q^n$ and look for the list of codewords in $\Gamma_q (L, G^{q-1})$ whose Hamming distance to $y$ is at most $\gamma n$, with $\gamma < \tau$.
One can apply Theorem \ref{soft} to $\code_L (\D, \A-q\E)$ with
$$
\forall i \in \{1, \ldots , n\},\ \forall \lambda \in \F_{q^m},\ w_{i,\lambda}\egaledef \left\{
  \begin{array}{ccl}
    1-\tau & \textrm{if} & \lambda=y_i \\
    \tau/(q-1) & \textrm{if} & \lambda \in \F_q\setminus \{y_i\}\\
    0 & \textrm{if} & \lambda \in \F_{q^m}\setminus\F_q
  \end{array}
\right. .
$$

\noindent From Theorem \ref{soft}, one can get the whole list of codewords of $\Gamma_q (L,G^{q-1})$ at distance at most $\gamma n$ from y in $poly(n,q, 1/\varepsilon)$ time provided

\begin{equation}\label{ineq1}
(1-\gamma)(1-\tau)+ \gamma \left(\frac{\tau}{q-1}\right) \geq
\sqrt{(1-\delta) \left({(1-\tau)}^2+\frac{\tau^2}{q-1} \right)} + \frac{\varepsilon}{n}(1-\tau).
\end{equation}

Consider the left hand side of the expected inequality (\ref{ineq1}) and use the assumption $\gamma < \tau$ together with the easily checkable fact $\tau q/(q-1)<1$. This yields

\begin{eqnarray}
(1-\gamma)(1-\tau)+\gamma \left(\frac{\tau}{q-1}\right) & = &  1-\tau -\gamma (1-\frac{\tau q}{q-1}) \\
 & > &  1-2 \tau +\frac{\tau^2 q}{q-1} \\
\label{rouge} & > & {(1-\tau)}^2 +\frac{\tau^2}{q-1}\cdot
\end{eqnarray}

\noindent On the other hand, an easy computation gives
\begin{equation}\label{tau}
{(1-\tau)}^2+ \frac{\tau^2}{q-1}=1-\delta.  
\end{equation}

\noindent Therefore, (\ref{rouge}) and (\ref{tau}) entail
$$
(1-\gamma)(1-\tau)+\gamma \left(\frac{\tau}{q-1}\right) > {(1-\tau)}^2 +\frac{\tau^2}{q-1}= \sqrt{(1-\delta)\left((1-\tau)^2+\frac{\tau^2}{q-1} \right)},
$$

\noindent which yields the expected inequality (\ref{ineq1}) provided $\varepsilon$ is small enough.

\end{proof}

\paragraph{A remark on Algebraic Geometric codes and one point codes}
In \cite{GuruswamiPhD}, when the author deals with AG codes, he only
considers one point codes, i.e. codes of the form $\code_L (\D,sP)$ where
$P$ is a single rational point an $s$ is an integer.
Therefore, Theorem \ref{soft} is actually proved (in
\cite{GuruswamiPhD}) only for one point codes and is applied in the
proof of Theorem \ref{list} to the code $\code_L (\D, \A-q\E)$ which is
actually not one point.

Fortunately, this fact does not matter since one can prove that any AG code on $\P^1$ is equivalent to a one point code. 
In the case of $\code_L (\D, \A-q\E)$, the equivalence can be
described explicitly.
Indeed, by definition of $\A$ and $\E$ we have
$
\A-q\E=(F')+(n-1)P_{\infty}-q(G)-q(\deg(G))P_{\infty}.
$
Set $d\egaledef \deg (G)$, then we get
$
\A-q\E=(F' G^q)+(n-1-qd)P_{\infty}.
$
Consequently, one proves easily that a codeword $c=(c_1, \ldots ,
c_n)$ is in $\code_L(\D,\A-q\E)$ if and only if $(\eta_1 c_1, \ldots , \eta_n
c_n) \in \code_L (\D, (n-1-qd)P_{\infty})$, where $\eta_i$'s are
defined by
$$
\forall i \in \{1, \ldots ,n\},\ \eta_i\egaledef
F'(\alpha_i)G^q(\alpha_i).
$$

%% file: DecodeClassicGoppa.tex
\subsection{List-decoding of general alternant codes}

In this subsection, we give a self-contained treatment of the proposed
list-decoding algorithm for alternant codes, up to the $q$-ary Johnson
bound, without the machinery of Algebraic Geometric codes. 

\begin{defn}
  Let $Q(X,Y)= \sum_{i,j}Q_{ij}X^iY^{j}\in\F_{q^m}[X,Y]$ be a
  bivariate polynomial and $s\geq0$ be an integer. We say that
  $Q(X,Y)$ has \deff{multiplicity at least $s$ at
  $(0,0)\in\F_{q^m}^2$} if and only if $Q_{ij} = 0$ for all $(i,j)$
  such that $i+j<s$.

  We say that $Q$ has \deff{multiplicity at least $s$ at point
  $(a,b)\in\F_{q^m}^2$} if and only if $Q(X+a,Y+b)$ has multiplicity
  $s$ at $(0,0)$. We denote this fact by
  $\mult\left(Q(X,Y),(a,b)\right)\geq s$.
\end{defn}
\begin{defn}
For $u,v\in \N$, the weighted degree $\wdeg_{u,v}(Q(x,y))$ of a
polynomial $Q(x,y)=\sum Q_{ij}x^i y^j$ is $\max\{ui+vy,\;
(i,j)\in\N\times\N\;|Q_{ij}\neq0\}$.\end{defn}

Let a $[n,\dimGRS,\dGRS]_{q^m}$ GRS$(L,B,\dimGRS)$ code be given and 
consider the corresponding alternant code $\code \egaledef
GRS_{|\mathbb{F}_q}$. We aim at list-decoding up to $ \gamma n$ errors,
where $\gamma$ is the relative list-decoding radius, which is
determined further.

Let $y \in\mathbb{F}_q^n$ be a received word. The main steps of the
algorithm are the following: Interpolation, Root-Finding,
Reconstruction.  Note that an auxiliary $s\in\N\setminus\{0\}$ is
needed, which is discussed further, and appropriately chosen. Now we
can sketch the algorithm. A pseudo-code is detailed further, see Algorithm \ref{algo:q-alt-decoding}.
\begin{enumerate}
\item Interpolation: Find $Q(X,Y)=\sum Q_i(X) Y^i\in\F_{q^m}[X,Y]$ such that
  \begin{enumerate}
  \item (non triviality) $Q(X,Y) \ne 0$;\label{cond:interp:non_zero}
\item (interpolation with varying multiplicities)\label{cond:interp:mult}
  \begin{itemize} \item mult$(Q(X,Y),(\alpha_i,y_i\beta_i^{-1})) \ge
  s(1-\gamma)$; \item mult$(Q(X,Y),(\alpha_i,z\beta_i^{-1})) \ge
  \frac{s\gamma}{q-1}$, for any $z \in \mathbb{F}_q \setminus
  \{y_i\}$;
\end{itemize}
  \item (weighted degree) $\wdeg_{1,\dimGRS} Q(X,Y)< sn\left((1-\gamma)^2 +
  \frac{\gamma^2}{q-1}\right)$; \label{cond:interp:wdeg}
\end{enumerate}

\item Root-Finding: Find all the factors $(Y-f(X))$ of $Q(X,Y)$, with $\deg f(X)<\dimGRS$;
\item Reconstruction: Compute the codewords associated to the $f(X)$'s
found in the Root-Finding step, using the evaluation map $\ev_{L,B}$.
Retain only those which are at distance at most $\gamma n$ from $y$.
\end{enumerate}

\input{PseudoCode}

\begin{lem}[\cite{GSListDecoding}]\label{Lem:Mult}
Let $u$ be an integer and $Q(X,Y)$ be a polynomial with multiplicity $s$ at
$(a,b)$. Then, for any $f(X)$ such that $f(a)=b$, one has $(X-a)^s
\mid Q(X,f(X))$.
\end{lem}

\begin{prop}
Let $y$ be the received word, and $Q(X,Y)$ satisfying
  conditions~\ref{cond:interp:non_zero},~\ref{cond:interp:mult},
  and~\ref{cond:interp:wdeg} above.  Let $f(X)$ be a polynomial such that $\deg
  f(X)<\dimGRS$ and accordingly, let $c=\ev_{L,B}(f(X))$. If
  $d(c,y)\leq \gamma n$, then $Q(X,f(X))=0$.
\end{prop}
\begin{proof}
Assume that $d(\ev(f(X)),y)=\theta n\leq \gamma n$. Set $I
\egaledef \left\{i,f(x_i)=y_i\beta_i^{-1}\right\}$ and $\overline I
\egaledef \left\{i,f(x_i)\neq y_i
\beta_i^{-1}\right\}$. Obviously we have $|I| = n(1-\theta)$ and
$|\overline{I}| = \theta n$.  Note that, from Lemma~\ref{Lem:Mult},
$Q(X,f(X))$ is divisible by
\[
\prod_{i\in I}(X-\alpha_i)^{\ceiling{s(1-\gamma)}}\times\prod_{j\in
  \overline I}(X-\alpha_j)^{\ceiling{s\gamma/(q-1)}},
\]
which is a polynomial of degree
$D=n(1-\theta)\ceiling{s(1-\gamma)}+n\theta\ceiling{\frac{s\gamma}{q-1}}$. This
degree is a decreasing function of $\theta$ for $\gamma <
\frac{q-1}{q}$, since it is an affine function of the variable
$\theta$, whose leading term is
\[
\theta n\left(\ceiling{s(1-\gamma)}+\ceiling{\frac{s\gamma}{q-1}}\right)<0.
\] The minimum is  reached for
$\theta=\gamma$, and is  greater than $sn\left((1-\gamma)^2 +
\frac{\gamma^2}{q-1}\right)$. Thus, $D\geq sn\left((1-\gamma)^2 +
\frac{\gamma^2}{q-1}\right)$.
On the other hand, the weighted degree condition imposed 
$Q(X,Y)$ implies that $\deg Q(X,f(X))< sn \left((1-\gamma) ^2 +
\frac{\gamma^2}{q-1})\right)$. Thus $Q(X,f(X))=0$.
\end{proof}

\begin{prop}
  Let $\deltaGRS=\frac \dGRS n$ be the relative minimum distance of a
  GRS code as above, defining an alternant subcode over $\F_q$.  Set
\begin{equation}
\label{tau:johnson:q}
  \tau \egaledef
  \frac{q-1}{q}\left(1-\sqrt{1-\frac{q}{q-1}\deltaGRS}\right).
\end{equation}
Then, for any $\gamma<\tau$, there exists $s$ large enough such that a polynomial $Q(X,Y)$, satisfying the three previous
constraints~\ref{cond:interp:non_zero},~\ref{cond:interp:mult},
and~\ref{cond:interp:wdeg}, always exists, whatever the received word.
\end{prop}
\begin{proof}
To make sure that, for every received word $y$, a non zero $Q(X,Y)$
exists, it is enough to prove that we have more indeterminates than equations in the
linear system given by \ref{cond:interp:mult},
and~\ref{cond:interp:wdeg}, since it is homogeneous. That is (see Appendix),
\begin{equation}\label{eq:ninc>neq}
\frac{s^2 n^2 ( (1-\gamma)^2 + \frac {\gamma ^2}{q-1})^2 }{2(k-1)}> \left(\binom{s(1-\gamma)+1}{2}+(q-1)\binom{s\frac\gamma{q-1}+1}{2}\right)n,
\end{equation}
which can be rewritten as,
\begin{align}
\left((1-\gamma) ^2 + \frac{\gamma ^2}{q-1} \right) ^2 >\pseudoR\left((1-\gamma) ^2 + \frac{\gamma ^2}{q-1} +\frac1s\right),
\end{align}
where $\pseudoR \egaledef \frac{\dimGRS-1}n \cdot$
Thus, we find that $\mu=\mu(\gamma)=(1-\gamma) ^2 + \frac{\gamma
^2}{q-1} $ must satisfy $
\mu^2 -\pseudoR\mu-\frac {\pseudoR}s>0$.
The roots of the equation  $\mu^2 -\pseudoR\mu-\frac {\pseudoR}s=0$ are 
\[
\mu_0=\frac{\pseudoR-\sqrt{\pseudoR^2+4\frac{\pseudoR} s}}2,\quad
\mu_1=\frac{\pseudoR+\sqrt{\pseudoR^2+4\frac{\pseudoR} s}}2 \cdot
\]
Note that the function $\mu(\gamma)$ is decreasing with
$\gamma\in[0,1-\frac1q]$, as shown on Fig~\ref{fig:mu} for the
particular case $q=2$.
\begin{figure}
\begin{center}
\includegraphics[scale=.25]{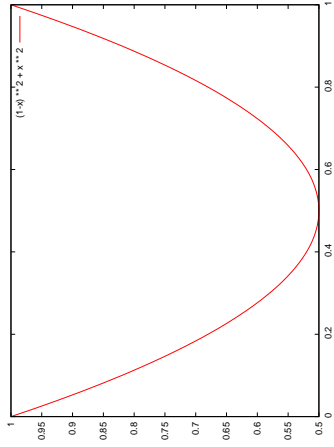}
\end{center}
\caption{Behaviour of the $\mu(\gamma)$ function, for $\gamma\in[0,1]$, with $q=2$.}
\label{fig:mu}
\end{figure}
Only $\mu_1$ is positive and thus we must have $\mu>\mu_1$, i.e.
\begin{align}
(1-\gamma)^2+\frac{\gamma ^2}{q-1} > \mu_1.
\end{align}
 Again, we have two roots for the equation $(1-\gamma)^2+\frac{\gamma
 ^2}{q-1}=\mu_1$, namely:
\begin{align}
\gamma_0&=\frac{q-1}q\left(1-\sqrt{1-\frac q{q-1}(1-\mu_1})\right)\\
\gamma_1&=\frac{q-1}q\left((1+\sqrt{1+\frac q{q-1}(1-\mu_1})\right).
\end{align}
Only $\gamma_0<\frac{q-1}q$, and thus we must have
\begin{align}\label{eq:tau(s)}
\gamma&<\gamma_0=\frac{q-1}q(1-\sqrt{1-\frac q{q-1}(1-\mu_1}).
\end{align}
Then, when $s\rightarrow \infty$, we have $\mu_1\rightarrow \pseudoR$, and we
get
\begin{align}
\gamma&<\tau=\frac{q-1}q\left(1-\sqrt{1-\frac q{q-1}(1-\pseudoR)}\right)
\end{align}
Using the fact that $\dimGRS-1=n-\dGRS$,
i.e\ $\pseudoR=1-\deltaGRS$, we get
\[
\tau= \frac{q-1}q\left(1-\sqrt{1-\frac q{q-1}\deltaGRS}\right),
\]
which the $q$-ary Johnson radius.
\end{proof}

The previous Proposition proves that this method enables to list-decode
any alternant code up to the $q$-ary Johnson bound. This bound
is better than the error correction capacities of the previous algorithms
\cite{GSListDecoding,BernsteinLDGC}. For the binary case, we plot in
Figure \ref{Fig:ErrorCapacity} the binary Johnson bound (weighted
multiplicities, this paper), the generic Johnson bound (straight
Guruswami-Sudan, or Bernstein), and the unambiguous decoding bound
(Patterson). As usually, the higher the normalised minimum
distance is, the better the Johnson bound is.

\begin{figure}[ht]
  \begin{center} \includegraphics[width=10cm]{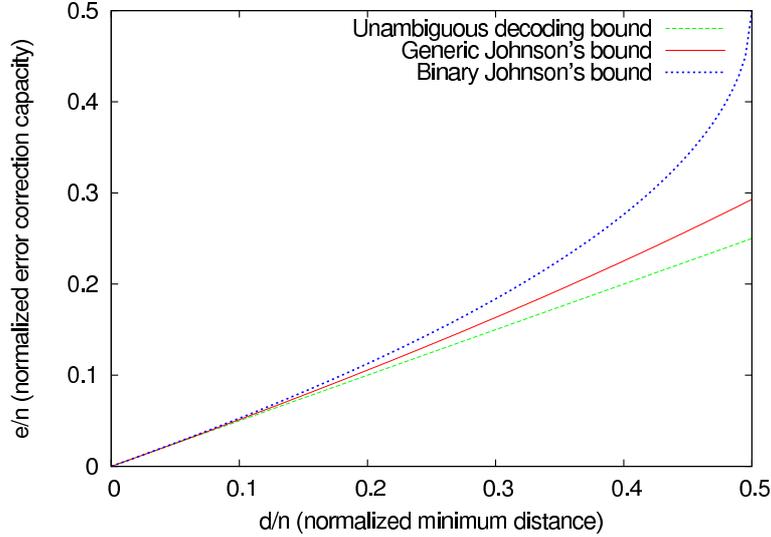}
    \caption{\label{Fig:ErrorCapacity} Comparison of the relative error
    capacities of different decoding algorithm for binary alternant
    codes --- This applies to binary square-free Goppa codes.} \end{center}
\end{figure}

\subsection{Complexity Analysis}
The main issue is to give explicitly how large the ``order of
multiplicity'' $s$ has to be, to approach closely the limit
correction radius,  $\tau(\deltaGRS)$ as given by~(\ref{tau:johnson:q}).
\begin{lem}
To list-decode up to a relative radius of $\gamma=(1-\varepsilon)\tau$,
it is enough to have an auxiliary multiplicity $s$ of size
$\bigO(\frac 1\varepsilon)$, where the constant in the big-$\bigO$
depends only on $q$ and on the pseudo-rate $\pseudoR=\frac {\dimGRS-1} n$ of the
GRS code.
\end{lem}
\begin{proof}
To get the dependency on $s$, we work out Equation (\ref{eq:tau(s)}). Let us denote
by $\gamma(s)$ the achievable relative correction radius for a given $s$. We have 
\begin{align}
\gamma      (s)
    &=
\frac{q-1}q
      \left(1-
        \sqrt{
	   1-
	   \frac q{q-1}
	    \left(
	         1-\mu_1(s)
	    \right)
	    }
      \right),
\end{align}
with $\mu_1(s)=\frac{\pseudoR+\sqrt{\pseudoR^2+4\frac{\pseudoR} s}}2$. We use that
$\sqrt{1+x}<1+\frac x2$, for all $x>0$. First, we have the bound:
\begin{align}
\mu_1(s)&=\frac{\pseudoR+\sqrt{\pseudoR^2+4\frac{\pseudoR} s}}2\\
    &= \frac \pseudoR2\left(1+\sqrt{1+\frac{4}{s\pseudoR}}\right)\\
    &\leq\frac \pseudoR2\left(1+\left(1+\frac{4}{2s\pseudoR}\right)\right)\\
    &=\pseudoR+\frac1s.
\end{align}
Now, calling $K_q(\pseudoR)$ the quantity $1-\frac q{q-1}(1-\pseudoR)$, we compute:
\begin{align}
\gamma(s)      
    &=
\frac{q-1}q
      \left(1-
        \sqrt{
	   1-
	   \frac q{q-1}
	    \left(
	         1-\mu_1(s)
	    \right)
	    }
      \right)\\
    &\geq \frac{q-1}q
      \left(1-
        \sqrt{
	   1-
	   \frac q{q-1}
	    \left(
	         1-\pseudoR-\frac 1s
	    \right)
	    }
      \right)
\\
&= \frac{q-1}q
      \left(1-
        \sqrt{
	   K_q(\pseudoR)+\frac q{q-1}\frac 1s
	    }
      \right)
\\
&=\frac{q-1}q
      \left(1-
        \sqrt{
	   K_q(\pseudoR)\left(1+\frac q{q-1}\frac 1s\frac 1{K_q(\pseudoR)}\right)
	    }
      \right)
\\
&\geq \frac{q-1}q
      \left(1-
        \sqrt{
	   K_q(\pseudoR)}\left(1+\frac 12\frac q{q-1}\frac 1s\frac 1{K_q(\pseudoR)}\right)
      \right)
\\
&=\frac{q-1}q
      \left(1-
        \sqrt{
	   K_q(\pseudoR)}-\sqrt{
	   K_q(\pseudoR)}\frac 12\frac q{q-1}\frac 1s\frac 1{K_q(\pseudoR)}
      \right)\\
&=\frac{q-1}q
      \left(1-
        \sqrt{
	   K_q(\pseudoR)}\right)-\frac 12\frac 1{s\sqrt{K_q(\pseudoR)}}\\
&=\tau-\frac 1{2s\sqrt{K_q(\pseudoR)}}\\
&=\tau\left(1-\frac 1{2s\tau\sqrt{K_q(\pseudoR)}}\right)
\end{align}
Thus, to reach a relative radius $\gamma=(1-\varepsilon)\tau$, it is
enough to take
\begin{equation}\label{eq:s(epsilon)}
s=\frac{1}{2\varepsilon\tau\sqrt{K_q(\pseudoR)}}=\bigO (\frac1\varepsilon ).
\end{equation}
\end{proof}
The most expensive computational step in these kinds of list-decoding
algorithms is the interpolation step, while root-finding is
surprisingly
cheaper~\cite{Augot-Pecquet:IEEE_IT2000,Roth-Ruckenstein:IEEE_IT2000}.
This is the reason why we focus on the complexity of this step. We assume the use
of the so-called Koetter algorithm~\cite{Koetter:THS1996} (see for
instance \cite{Trifonov:IEEE_IT2010,Koetter-Ma-Vardy:ARXIV2010} for a
recent exposition), to compute the complexity of our method. It is in
general admitted that this algorithm, which also may help in various
interpolation problems, has complexity $\bigO(lC^2)$, where $l$ is the
$Y$-degree of the $Q(X,Y)$ polynomial, and $C$ is the number of linear
equations given by the interpolation constraints. It can be seen as an
instance of the Buchberger-M\"oller
problem~\cite{Moller:1982:CMP:646656.700243}.

\begin{cor}
  The proposed list-decoding runs in
  \begin{equation}
  \bigO (\frac 1{\varepsilon^{-5}}n^ 2)
  \end{equation}
  field operations
  to list-decode up to $(1-\varepsilon) \tau\cdot n$ errors, where
  the constant in the big-$\bigO$ depends only on $q$ and the
  pseudo-rate $\pseudoR$.
\end{cor}

\begin{proof}
Assume that we would like to decode up to $n\gamma = n(1
-\varepsilon)\tau$. The number of equations given by the interpolation
conditions can be seen to be $\bigO(ns^2)$ (see Equation (\ref{eq:ninc>neq})). Now, the list size $\ell$ is
bounded above by the $Y$-degree of the interpolation polynomial $Q(X,Y)$,
which is at most
\begin{equation}\label{eq:list-size}
\frac{sn((1-\gamma)^2 + \frac{\gamma^2}{q-1})}{k-1}=\bigO(s),
\end{equation}
for fixed $\pseudoR=\frac{k-1}n$. Fitting $s=\bigO(\frac1\varepsilon)$, we
conclude that this method runs in $\bigO(n^2\varepsilon^{-5})$.

Regarding the Root-Finding step, one can
use~\cite{Roth-Ruckenstein:IEEE_IT2000}, where an algorithm
of complexity $\bigO(\ell^3 k^2)$ is proposed, assuming $q$ is
small. Indeed, classical bivariate factorisation or root finding
algorithms rely on a univariate root-finding step, which is not
deterministic polynomially in the size of its input, when $q$
grows. But our interest is for small $q$, i.e.\ 2 or 3, and we get
$\bigO(s^3 n^2)$, which is less than the cost of the interpolation step.
\end{proof}

\begin{cor}\label{Lem:GeneralComplexity}
To reach the non relative $q$-ary Johnson radius: 
\[
\left\lceil
  \frac{q-1}{q}n\left(1-\sqrt{1-\frac{q}{q-1} \frac{\dGRS}{n}}\right)
  -1\right\rceil,
\] it is enough to have $s=\bigO(\frac1n)$.  Then, the
  number of field operations, is
  $$
  \bigO (n^7),
  $$
where the constant in the
  big-$\bigO$ only depends on  $q$ and  $\pseudoR$.

\end{cor}

\begin{proof}
It is enough to consider that 
\[
\left\lceil
  \frac{q-1}{q}n\left(1-\sqrt{1-\frac{q}{q-1} \frac{\dGRS}{n}}\right)
  -1\right\rceil=n\tau(1-\varepsilon),
\]
with $\varepsilon=\bigO(\frac1n)$.
\end{proof}

\subsection{Application to classical binary Goppa codes}

The most obvious application of this algorithm is the binary Goppa
codes defined with a square-free polynomial $G$ (we do not detail the
result for the general $q$-ary case, which is less relevant in
practice). Indeed, since we have
\[
\Gamma_2(L,G)=\Gamma_2(L,G^2),
\] 
both codes benefit at least from the dimension of $\Gamma_2(L,G)$ and
the distance of $\Gamma_2(L,G^2)$.  Thus, if $\deg G=t$, we have
the  decoding radii given in Table~\ref{tab:goppa}.
In addition, we compare in Table \ref{Tab:Compare} the different decoding radii
for practical values.

\begin{table}
\begin{center}
\begin{tabular}{|c|c|c|c|}
\hline
Generic list-decoding & Bernstein &Binary list-decoding\\
\hline
\null & & \\
$\ceiling{n-\sqrt{n(n-2t)}}-1$ & $n-\sqrt{n(n-2t-2)}$&
$\ceiling{\frac12\left(n-\sqrt{n(n-4t-2}\right)}-1$\\
 &  & \\
\hline
\end{tabular}
\end{center}
\caption{Comparison of the claimed decoding radii in terms of $t$,
the degree of the square-free polynomial $G(X)$ using to construct the
Goppa code.}\label{tab:goppa}
\end{table}


\begin{algorithm}
\caption{List-decoding of binary Goppa codes}
\label{algo:goppa-list-decode}
\begin{algorithmic}[1]
\REQUIRE\ 

$L=(\alpha_1,\dots,\alpha_n)$ \\
A Goppa polynomial $G$, square-free\\
The corresponding Goppa code $C=\Gamma_2(L,G)$  \\
  The associated evaluation map $\ev$\\
The relative decoding radius $\gamma=(1-\varepsilon)\tau$\\
The received word $y\in\F_q^n$\\
\STATE\ View $\Gamma_2(L,G)$ as $\Gamma_2(L,G^2)$
\STATE\ Consider the Generalised Reed-Solomon  code $GRS_{q^m}(L,B,k)$ above $\Gamma_2(L,G^2)$
\STATE\ Use Algorithm~\ref{algo:q-alt-decoding} to find all the codewords at  distance $\gamma n$ of $y$
\end{algorithmic}
\end{algorithm}

\begin{table}[ht]
  \begin{center}
    \begin{tabular}{|c|c|c||c|c|c|c|}
      \hline
      $n$ & $k$ & $t$ & Guruswami-Sudan & Bernstein & Binary list-decoding
      \\
      \hline
      \hline
      16 & 4 & 3  & 4 & 4 & 5\\
      \hline
      \hline
      32 & 2 & 6 & 7 & 8 & 9\\
      \hline
      \hline
      64 & 16 & 8 & 9 & 9 & 10\\
      \hline
      64 & 4 & 10 & 11 & 12 & 13\\
      \hline
      \hline
      128 & 23 & 15 & 16 & 17 & 18\\
      \hline
      128 & 2 & 18 & 20 & 20 & 22\\
      \hline
      \hline
      256 & 48 & 26 & 28 & 28 & 30\\
      \hline
      256 & 8 & 31 & 33 & 34 & 36\\
      \hline
      \hline
      512 & 197 & 35 &  36 & 37 & 38\\
      \hline
      512 & 107 & 45 &  47 & 48 & 50\\
      \hline
      512 & 17 & 55 &  58 & 59 & 63\\
      \hline
      \hline
      1024 & 524 & 50  & 51 & 52 & 53\\
      \hline
      1024 & 424 & 60  & 62 & 62 & 74\\
      \hline
      1024 & 324 & 70  & 73 & 73 & 76\\
      \hline
      1024 & 224 & 80  & 83 & 84 & 88\\
      \hline
      1024 & 124 & 90  & 94 & 95 & 100\\
      \hline
      \hline
      2048 & 948 & 100  & 103 & 103 & 105\\ 
      \hline
      2048 & 728 & 120  & 124 & 124 & 128\\ 
      \hline
      2048 & 508 & 140  & 145 & 146 & 151\\ 
      \hline
      2048 & 398 & 150  & 156 & 157 & 163\\ 
      \hline
      2048 & 288 & 160  & 167 & 167 & 175\\ 
      \hline
      2048 & 178 & 170  & 178 & 178 & 187\\ 
      \hline
    \end{tabular}
    \caption{\label{Tab:Compare}Comparison of the error capacities of
      different decoding algorithms for square-free binary Goppa
      codes, with respect to the length $n$, the dimension $k$, and
      the degree $t$ of $G(X)$.}
  \end{center}
\end{table}


%% file: PseudoCode.tex
\begin{algorithm}
\caption{List decoding of alternant codes up to the
$q$-ary Johnson bound}
\label{algo:q-alt-decoding}
\begin{algorithmic}[1]
\renewcommand{\algorithmicrequire}{\textbf{Subroutine:}}
\REQUIRE\

{\tt Interpolation}({\tt constraints},$k-1$) finds a polynomial $Q(X,Y)$ satisyfing the constraints~\ref{cond:interp:non_zero},~\ref{cond:interp:mult},
  and~\ref{cond:interp:wdeg}
\renewcommand{\algorithmicrequire}{\textbf{Input:}}
\REQUIRE\

$L=(\alpha_1,\dots,\alpha_n)$\\
$B= (\beta_1, \ldots , \beta_n)$\\
  The associated evaluation map $\ev$\\
 $\dimGRS$\\
$C$, the alternant code $GRS((\alpha_i),(\beta_i),\dimGRS)_{|\mathbb{F}_q}$
  \\
The relative decoding radius $\gamma=(1-\epsilon)\tau$\\ 
The received word $y\in\F_q^n$\\
\renewcommand{\algorithmicrequire}{\textbf{Output:}}
\REQUIRE\
The list of codewords $c\in C$ such that $d(c,y)\leq \gamma n$
\STATE\ $s, \ell \longleftarrow$ Parameters($n,\dimGRS,
  \epsilon$), according to Equations (\ref{eq:s(epsilon)}) and (\ref{eq:list-size}).
\STATE {\tt constraints} $\gets[]$
\FOR {$i=1$ to $n$}
\FOR{ $z\in\F_q$}
\IF{$z=y_i$}
\STATE {\tt constraints} $\gets$ {\tt constraints} $ \cup 
       \left\{(\alpha_i,z\beta_i^{-1}),\lceil
       s(1-\gamma)\rceil\right\}$;
\ELSE
\STATE
       {\tt constraints} $\gets$ {\tt constraints} $ \cup 
       \left\{(\alpha_i,z\beta_i^{-1}),\lceil
       \frac{s\gamma}{q-1}\rceil\right\}$;
\ENDIF
\ENDFOR
\ENDFOR
\STATE $Q(X,Y)\gets$ {\tt Interpolation}({\tt constraints},$k-1$)
\STATE $F\gets\{ f(X)\;\mid (Y-f(X))\mid Q(X,Y)\}$
\STATE {\bf Return} $\{c=f(X)\;\mid f(X)\in F\text{\bf\ and } \deg f(X)< k-1 \text{\bf\ and } d(c,y)\leq \gamma n\}$
\end{algorithmic}
\end{algorithm}

%% file: Appendix.tex
\section{The number of unknowns}

\begin{prop}
Let $Q(X,Y)\in\F[X,Y]$ be a bivariate polynomial such that
$\wdeg_{1,k-1} Q(X,Y)< D$. Then, the number $N_{k-1,D}$ of nonzero
coefficients of $Q(X,Y)$ is larger than or equal to
\begin{equation}
\frac{D^2}{2(k-1)} \cdot
\end{equation}
\end{prop}
\begin{proof}
Let us write
\[
Q(X,Y)=\sum_{i=0}^{\ell} Q_i Y^i,
\]
with $\deg Q_j(X)<D-(k-1)j$, and  $\ell$ maximal such that $(k-1)\ell<D$. Then
\begin{align}
N_{k-1,D}=&\sum_{i=0}^\ell\left(D-(k-1)i\right)\\
   =&(\ell+1)D-(k-1)\frac{\ell(\ell+1)}2\\
   =&(\ell+1)\left(D-(k-1)\frac l 2 \right)\\
   >&(\ell+1)\left(D-\frac D 2 \right)\\
   =&(\ell+1)\frac D 2 \\
   \geq &\frac {D^2}{2(k-1)}\cdot
\end{align}
\end{proof}

\section{The number of interpolation constraints}
The following proposition can be found in any book of computer
algebra, for instance~\cite{Cox-Little-Oshea:IVAA1992}.
\begin{prop}
Let $Q(X,Y)\in \F[X,Y]$ be a bivariate polynomial. The number of terms
of degree at most $s$ in $Q(X,Y)$ is $\binom {s+1}2$.
\end{prop}
\begin{cor}
The condition $\mult(Q(X,Y)),(a,b))\geq s$ imposes $\binom {s+1}2$
linear equations on the coefficients of $Q(X,Y)$. Let
$(a_i,b_i)\in\F^2$ be points and $s_i\in\N$ be multiplicities,
$i\in\intervalle1n$. Then the number of linear equations imposed by
the conditions
\begin{equation}
\mult(Q(X,Y),(a_i,b_i))\geq s_i,\quad i\in\intervalle1n,
\end{equation}
is 
\begin{equation}
\sum_{i=1}^n \binom{s_i+1}2.
\end{equation}
\end{cor}